\documentclass[a4paper,UKenglish,cleveref, autoref, thm-restate]{lipics-v2021}
\usepackage[utf8]{inputenc}
\usepackage{commath}
\usepackage{mathtools}
\DeclarePairedDelimiter\floor{\lfloor}{\rfloor}
\DeclarePairedDelimiter\ceil{\lceil}{\rceil}

\usepackage[utf8]{inputenc}
\usepackage{microtype}
\usepackage{amsfonts}
\usepackage{amssymb}
\usepackage{amsmath}
\usepackage{amsthm}
\usepackage{thm-restate}
\usepackage{comment}
\usepackage{xspace}
\usepackage{enumitem}
\usepackage{algorithm}
\usepackage{algpseudocode}
\newtheorem{problem}[theorem]{Problem}
\newcommand{\argmin}{\operatorname{argmin}}

\def\Cdo {CDO}
\def\CDO {CDO}
\def\MCDOCH {MCDOCH}
\def\AMCDOCH {AMCDOCH}
\def\MCDO {MCDO}
\def\NNS {NNS}

\def\RNNS {RNNS}
\def\MPMP {$(min,+)$-matrix product}

\def\ACDO {$\text{A\Cdo{}}$}

\def\DRMQ {2DRMQ}
\def\\max(S) {$\max(S)$}
\newcommand{\ThreeSUM}{\textsf{3SUM}}
\newcommand{\APSP}{\textsf{APSP}}
\newcommand{\SetDisjointness}{\textsf{SetDisjointness}}

\newcommand{\ignore}[1]{}

\title{Color Distance Oracles and Snippets: Separation Between Exact and Approximate Solutions}

\titlerunning{Color Distance Oracles and Snippets: Separation Between Exact and Approximate Solutions}
\author{Noam Horowicz}{Bar Ilan University, Israel}{noamhorowicz@gmail.com}{https://orcid.org/0009-0005-6909-234X}{}

\author{Tsvi Kopelowitz}{Bar Ilan University, Israel}{kopelot.biu@gmail.com}{https://orcid.org/0000-0002-3525-8314}{}

\funding{Supported by a BSF grant 2018364, and by an ERC grant MPM under the EU's Horizon 2020 Research and Innovation Programme (grant no. 683064).}

\authorrunning{N. Horowicz and T. Kopelowitz}


\Copyright{Noam Horowicz and Tsvi Kopelowitz}
\category{}

\nolinenumbers
\hideLIPIcs

\begin{document}

    \maketitle
\begin{abstract}
    In the \emph{snippets} problem, the goal is to preprocess a text $T$ so that given two pattern queries, $P_1$ and $P_2$, one can quickly locate the occurrences of the two patterns in $T$
that are closest to each other, or report the distance between these occurrences.
Kopelowitz and Krauthgamer [CPM2016] showed upper bound tradeoffs and conditional lower bounds tradeoffs for the snippets problem, by utilizing connections between the snippets problem and the problem of constructing a color distance oracle (CDO), which is a data structure that preprocess a set of points with associated colors so that given two colors $c$ and $c'$ one can quickly find the (distance between the) closest pair of points where one has color $c$ and the other has color $c'$.
However, the existing upper bound and lower bound curves are not tight.

Inspired by recent advances by Kopelowitz and Vassilevska-Williams [ICALP2020] regarding tradeoff curves for Set-disjointness data structures, in this paper we introduce new conditionally optimal algorithms for a $(1+\varepsilon)$ approximation version of the snippets problem and a $(1+\varepsilon)$ approximation version of the CDO problem, by applying fast matrix multiplication.
For example, for CDO on $n$ points in an array, if the preprocessing time is $\tilde{O}(n^a)$ and the query time is $\tilde{O}(n^b)$ then, assuming that $\omega=2$ (where $\omega$ is the  exponent of $n$ in the runtime of the fastest matrix multiplication algorithm on two squared matrices of size $n\times n$), we show that approximate CDO can be solved with the following tradeoff \[\begin{dcases}
        a +  2b = 2 &\text{if } 0 \leq  b \leq \frac1 3 \\
        2a + b = 3  &\text{if } \frac13\leq b \leq 1.
    \end{dcases}\]

Moreover, we prove that for exact CDO on points in an array,  the algorithm of Kopelowitz and Krauthgamer [CPM2016], which obtains a tradeoff of $a+b =2$, is essentially optimal assuming that the \emph{strong all-pairs shortest paths} hypothesis holds for randomized algorithms.
Thus, we demonstrate that the exact version of CDO is strictly harder than the approximate version.
Moreover, this separation carries over to the snippets problem.

\end{abstract}
\newpage

\section{Introduction}\label{sec:intro}
In the \emph{snippets} problem, introduced by Kopelowitz and Krauthgamer~\cite{KK16}, the goal is to preprocess a text $T$ so that given two pattern queries, $P_1$ and $P_2$, one can quickly locate the locations of the two patterns in $T$
that are closest to each other, or report the distance between these locations.
The snippets problem is motivated by many common text indexing applications, such as  searching a corpus
of documents for two query keywords.
Often, the relevance of a document to the query is  measured by
the proximity of the two keywords within the document.
The term \emph{snippet} is derived from  search engines often  providing  each result with a snippet of text from the document  with  the two keywords  close to each other.

Kopelowitz and Krauthgamer~\cite{KK16} designed a tradeoff algorithm for the snippets problem based on a connection to a problem on colored points in metric spaces, which we define next.

\paragraph*{Colored points and color distances. }
Let $\mathbf{M}$ be a metric space with distance function $d(\cdot,\cdot$).
    Let $S \subseteq \mathbf{M}$ be a set of points where $\abs{S} = n$.
    Let $\mathcal{C}$ be a set of \emph{colors}.
    For sake of convenience we assume that $\mathcal C = [|\mathcal C|]$.
    Each point $p$ $\in$ $S$ has an associated color $c_p \in \mathcal{C}$.
    For a color $c\in \mathcal C$, let $P_S(c)=\{p\in S | c_p = c\}$.
    When $S$ is clear from context we abuse notation and denote $P(c) = P_S(c)$.

For sets of points $A,B\subset \mathbf M$ and a point $p\in M$, let $d(p,A) = \min_{p'\in A}\{d(p,p')\}$ and let $d(A,B) = \min_{p'\in A}\{d(p',B)\} = \min_{p'\in A, \hat p\in B}\{d(p',\hat p)\}$.
We extend the definition of $d$
to inputs of colors as follows.
For colors $c,c'\in \mathcal C$ and point $p\in \mathbf{M}$, let $d(p,c) = d(p,P(c))$ and let $\delta_{c,c'} = d(c,c') = d(P(c),P(c'))$.
We say that $\delta_{c,c'}$ is the \emph{color distance} between $c$ and $c'$.
A natural problem is to construct a \emph{color distance oracle} (\CDO), which is defined as follows.

    \begin{problem}
    [The Color Distance Oracle (\Cdo{}) problem~\cite{KK16}]\label{prb:CDO} Given a set $S\subseteq \mathbf M$ of $n$ colored points with colors from $\mathcal C$, preprocess $S$ to support the following query: given $c, c'$ $\in$ $\mathcal{C}$, return $\delta_{c,c'}$.
    \end{problem}

    In this paper we consider  metrics defined by locations in an array of size $s$, and so the metric is the set\footnote{Throughout this paper for a positive integer $k$ we denote $[k] = \{1,2,,\ldots,k\}$.} $[s]$, and the distance function of two points $p$ and $q$ is $d(p,q) = \abs{p-q}$.

\paragraph*{Multi-colored points and color hierarchies.}
A natural generalization of a colored set is to allow each  $p\in S$ to be associated with a nonempty set of colors $c(p) \subseteq \mathcal C$, and in such a case $S$ is said to be \emph{multi-colored}.
This version of the \CDO{} problem is called the \emph{multi-color distance oracle} (\MCDO{}) problem.
In general, representing $c(p)$ costs $\Theta (|c(p)|)$ space by explicitly listing all colors in $c(p)$.
Thus, the input size is $n=\sum_{p\in S}|c(p)|$.

Nevertheless, there are interesting cases in which the lists of colors for each point are not required to be given explicitly.
One such example, which we focus on, is when for every two colors $c, c'\in \mathcal C$, either one of the sets $P(c)$ and $P(c')$ contains the other, or the two sets are disjoint.
In such a case, we say that $\mathcal C$ has a \emph{color hierarchy} with respect to $S$
(a formal terminology is that $\{P(c)\}_{c\in \mathcal C}$ is a laminar family), represented by a rooted forest (i.e., each tree has a root) $T_S$ of size $O(|\mathcal C|)$.
Each color $c$ is associated with a vertex $u_c$ in $T_S$, such that the descendants of $u_c$ are exactly all the vertices $u_{c'}$ whose color $c'$ satisfies $P(c')\subseteq P(c)$.
We convert the forest $T_S$ to a rooted tree by adding a dummy root vertex and making the dummy root the parent of all of the roots of the trees in the forest.

    With the aid of $T_S$, a multi-colored set $S$ that admits a color hierarchy can be represented using only $O(n+|\mathcal C|)$ machine words, because it suffices to store $T_S$ and associate with each point $p$ just one color $c$
(the color with the lowest corresponding vertex in $T_S$);
the other colors of $p$ are implicit from $T_S$
(the colors on the path from $u_c$ to the root of $T_S$).
Notice that $|\mathcal C| > n$ implies that there are at least two  colors with the same point set.
Thus, we make the simplifying assumption that all colors have different point sets\footnote{If we happen to have several colors with the same point set, we ignore all of them but one during preprocessing, and use a straightforward mapping from the original coloring to the reduced coloring during query time.}, and so the input size is $O(n)$.

Thus, a natural extension of Problem~\ref{prb:CDO} is the following.

\begin{problem}[The Multi-Color Distance Oracle with a Color Hierarchy (\MCDOCH{}) problem~\cite{KK16}]
Given a set $S\subseteq \mathbf M$ of $n$ multi-colored points with colors from $\mathcal C$ which form a color hierarchy with respect to $S$, preprocess $S$ to support the following query: given  $c, c'$ $\in$ $\mathcal{C}$, return $\delta_{c,c'}$.
\end{problem}

Let $a$ and $b$ be real numbers such that the preprocessing and query times of a \MCDOCH{} algorithm are $\Tilde{O}(n^a)$\footnote{Throughout this paper we use the notation $\tilde O(\dot)$ to suppress sub-polynomial factors.} and $\Tilde{O}(n^b)$, respectively.
A tradeoff algorithm for the  \MCDOCH{} problem was presented in~\cite{KK16}, where given parameter $1\le \tau \le n$, the query time is $\Tilde{O}(n^b) = \tilde O(\tau)$ and the preprocessing time is $\tilde{O}(n^a) = \tilde O(n^2/\tau)  = \tilde O(n^{2-b})$.
When ignoring sub-polynomial factors, this upper bound tradeoff is given by $a+b = 2$ where $0\le b \le 1$ and $1\le a \le 2$.
Notice that the same tradeoff applies to the \CDO{} problem.

In addition to the tradeoff algorithms, ~\cite{KK16} proved a  \ThreeSUM{} based conditional lower bound (CLB) tradeoff\footnote{The CLB tradeoff  stated here is straightforward to derive from the statement of Theorem 4 in~\cite{KK16}.} of $a+2b \ge 2$
for the \CDO{} problem (and hence for the \MCDOCH{} problem), even when $\mathbf M$ is defined by locations in an array.
Their CLB  holds also for approximate versions of the \CDO{} problem (see Theorem 8 in~\cite{KK16}), where for a fixed $\alpha >1$, the answer to a color distance query between $c$ and $c'$ is required to be between $\delta_{c,c'}$ and $\alpha \cdot \delta_{c,c'}$.

\paragraph*{Solving the Snippets problem.}
By designing a reduction from the snippets problem to the \MCDOCH{} problem, \cite{KK16} were able to design an algorithm for the snippets problem with  essentially the same tradeoff:
for a chosen parameter $1\le \tau \le |T|$ their snippets algorithm  uses $\tilde O (|T|^2/\tau)$ preprocessing time and answers queries in $\tilde O(|P_1| + |P_2| + \tau)$ time.
The  main idea in the reduction is to construct the suffix tree of $T$, assign a color to each vertex in the suffix tree, and for each leaf $\ell$ in the suffix tree, the corresponding location in $T$ is assigned all of the colors on the path from $\ell$ to the suffix tree root.
It is straightforward to see that the colors define a hierarchy and that $T_S$ is exactly the suffix tree.

\paragraph*{A complexity gap.}
The upper bound curve of $a+b = 2$ and the CLB tradeoff of $a+2b\ge 2$ form a complexity gap for the snippets and (approximate) \CDO{} problems, which was left open by~\cite{KK16}.
Moreover,  based on the combinatorial Boolean Matrix Multiplication (BMM) hypothesis, \cite[Theorem 8]{KK16} implies that any algorithm that beats the $a+b = 2$ curve, even for approximate versions, must use non-combinatorial techniques, such as fast matrix multiplication (FMM) algorithms.

We remark that the CLB tradeoff given in~\cite{KK16} is based on  a reduction from a \SetDisjointness{} problem to the (approximate) \CDO{} problem, and at the time \cite{KK16} was published,  the \SetDisjointness{} problem was known to exhibit the same upper bound tradeoff and \ThreeSUM{} based CLB tradeoff~\cite{KPP16}.
However, recently Kopelowitz and Vassilevska-Williams~\cite{KVW20} closed this complexity gap by, among other ideas, using FMM techniques.
Thus, a natural question to ask is whether FMM can assist in obtaining improved algorithms for \CDO{} problems.

\paragraph*{Our results}
We  close the complexity gap for the $(1+\varepsilon)$ approximation versions of \CDO{} and  $\MCDOCH{}$, where $\mathbf M$ is defined by locations in an array.
Formally, the problems are defined as follows.

    \begin{problem}
    [The $(1+\varepsilon)$-Approximate Color Distance Oracle (\ACDO{}) problem on an array]
    Let  $\mathbf M$ be a metric defined by locations in an array of size $n$.
    Given a set $S \subset \mathbf M$ of $n$ colored points with colors from $\mathcal C$ and a fixed real $0\le \varepsilon \le 1$, preprocess $S$ to support the following query: given two colors $c, c'$ $\in$ $\mathcal{C}$, return  a value $\hat \delta$ such that  $\delta_{c,c'} \le \hat \delta \le (1+\varepsilon)\delta_{c,c'}$.
    \end{problem}

\begin{problem}[The $(1+\varepsilon)$-Approximate Multi-Color Distance Oracle  problem with a Color Hierarchy (\AMCDOCH{}) on an array]

    Let  $\mathbf M$ be a metric defined by locations in an array of size $\Theta(n)$.
    Given a set $S\subseteq \mathbf M$ of $n$ multi-colored points with colors from $\mathcal C$ which form a color hierarchy with respect to $S$, and a fixed real $0\le \varepsilon \le 1$, preprocess $S$ to support the following query: given two colors $c, c'$ $\in$ $\mathcal{C}$, return  a value $\hat \delta$ such that  $\delta_{c,c'} \le \hat \delta \le (1+\varepsilon)\delta_{c,c'}$.

\end{problem}

Our first main result is a new tradeoff algorithm for \ACDO{}, which is stated by the following theorem.
Notice that the time complexities presented are dependent on $\omega\ge 2$, which is the exponent of $n$ in the runtime of the fastest FMM algorithm on two squared matrices of size $n\times n$.
The currently best upper bound on $\omega$ is $\omega < 2.371339$ given by~\cite{ADWVXZ25}.
However, for clarity, we choose to express our tradeoffs in terms of $\omega$.

\begin{theorem}\label{thm:ACDO}
For any real $0\le b\le 1$, there exists an \ACDO{} algorithm with preprocessing time $\Tilde{O}(n^a)$ and query time $\tilde{O}(n^b)$ where
\[\begin{dcases}
        a +  \frac{2}{\omega - 1}b = 2 &\text{if } 0 \leq  b \leq \frac{\omega -1}{\omega + 1} \\
        \frac{2}{\omega - 1}a + b = {\frac{\omega+1}{\omega-1}}  &\text{if } \frac{\omega -1}{\omega + 1} \leq b \leq 1.
    \end{dcases}\]

\end{theorem}

\begin{theorem}\label{thm: HACDO by ACDO}
For any real $0\le b\le 1$, there exists an \AMCDOCH{} algorithm with preprocessing time $\tilde{O}(n^a)$ and query time $\tilde{O}(n^b)$ where
\[\begin{dcases}
        a +  \frac{2}{\omega - 1}b = 2 &\text{if } 0 \leq  b \leq \frac{\omega -1}{\omega + 1} \\
        \frac{2}{\omega - 1}a + b = {\frac{\omega+1}{\omega-1}}  &\text{if } \frac{\omega -1}{\omega + 1} \leq b \leq 1.
    \end{dcases}\]
\end{theorem}

We remark that it is straightforward to  adapt our \ACDO{} and \AMCDOCH{} algorithms to  return the two points (one of each color) that define the distance being returned. See discussion in~\ref{app:return points}.
Thus, by combining Theorem~\ref{thm: HACDO by ACDO} (assuming that the points are returned) with the reduction of the snippets problem to the \MCDOCH{} problem given in~\cite{KK16}, we obtain the following tradeoff for a $1+\varepsilon$ approximation version of the snippets problem.
\begin{theorem}\label{thm:snippets}
For any fixed $0< \varepsilon \le 1$, and $1\le a \le 2$, there exists an algorithm that preprocesses a text $T$ in $\tilde O(|T|^a)$ time such that given two pattern strings $P_1$ and $P_2$ where the distance between the  closest occurrences of the two patterns in $T$ is $\delta$, the algorithm returns a value $\hat \delta$ such that $\delta \le \hat \delta \le (1+\varepsilon)\delta$ in $\tilde O(|P_1|+|P_2| + \abs{T}^b)$ time, where
\[\begin{dcases}
        a +  \frac{2}{\omega - 1}b = 2 &\text{if } 0 \leq  b \leq \frac{\omega -1}{\omega + 1} \\
        \frac{2}{\omega - 1}a + b = {\frac{\omega+1}{\omega-1}}  &\text{if } \frac{\omega -1}{\omega + 1} \leq b \leq 1.
    \end{dcases}\]
\end{theorem}

\paragraph*{A complexity separation between exact and approximate solutions.}
We remark that by combining the reduction given in~\cite[Theorem 8]{KK16} from \SetDisjointness{} to the \CDO{} problem with the CLBs for \SetDisjointness{} given in~\cite[Theorem 6]{KVW20} for the case of $\omega =2$ (which some researchers believe should be obtainable), Theorems~\ref{thm:ACDO},~\ref{thm: HACDO by ACDO} and~\ref{thm:snippets} are all conditionally optimal\footnote{For Theorems~\ref{thm:ACDO} and~\ref{thm: HACDO by ACDO}, the optimality follows directly from applying \cite[Theorem 8]{KK16} and~\cite[Theorem 6]{KVW20}.
For Theorem~\ref{thm:snippets}, the optimality follows since an approximate solution for the snippets problem can be used to solve \ACDO{} on an array as follows: treat each color as a character, and then the array becomes a string. Preprocess the string using the snippets algorithm, so that given an \ACDO{} query on two colors $c,\hat c$, query the snippets algorithm with patterns $P=c$ and $\hat P= \hat c$.
Thus, the \ACDO{} lower bound applies to approximate snippets.
} (up to subpolynomial factors).
A natural question to ask is whether one can design exact algorithms with the same tradeoff bounds. We provide evidence that this is not possible, based on a \emph{strong} version of the popular all-pairs shortest paths (\APSP{}) conjecture~\cite{CVWX23}, thereby demonstrating a separation between \CDO{} and \ACDO{}.

The Strong-\APSP{} hypothesis, introduced by Chan, Vassilevska-Williams and Xu~\cite{CVWX23}, states that for a graph with $\hat n$ vertices\footnote{We use $\hat n$ to differentiate from $n=|S|$ in the various \CDO{} problems.}, even if all of the edge weights  are in\footnote{Actually, in \cite{CVWX23} the weights are bounded by $\hat n^{3-\omega}$, but since $\omega \ge 2$ we can  bound the weights by $\hat n$.} $[\hat n]$, \APSP{} still does not have a truly subcubic algorithm.
A closely related problem to $\APSP{}$ is the \emph{\MPMP{}}  problem, where the input is two $\hat n$ by $\hat n$ matrices  $A = \{a_{i,j}\}$ and $B =\{b_{i,j}\} $, and the output is the  matrix $D=\{d_{i,j}\}$ where $d_{i, j} = \min_{1\le k\le \hat n}(a_{i, k}+b_{k, j})$.
Shoshan and Zwick~\cite{SZ99} showed that solving \APSP{} with weights $[M]$ is equivalent (in the sense of having the same runtime) to solving \MPMP{} with entries in $[O(M)]$.
Thus, the Strong-\APSP{} hypothesis implies that there is no  truly subcubic time \MPMP{} algorithm even when all of the entries are in $[\hat n]$.

By reducing \MPMP{} with entries in $[\hat n]$ to \MCDO{}, we are able to prove the following CLB in Section~\ref{sec:minplus LB}, which matches the upper bound given in~\cite{KK16}.

\begin{theorem}\label{thm:minplus LB}
    Assuming the Strong-\APSP{} conjecture, any algorithm for \MCDO{} on $n$ points in an array of size $O(n)$ with $\Tilde{O}(n^a)$ preprocessing time and $\Tilde{O}(n^b)$ query time, respectively, must obey $a+b \ge 2$.
\end{theorem}

Moreover, we are able to leverage the ideas used in the proof of Theorem~\ref{thm:minplus LB} to reduce \MPMP{} with entries in $[\hat n]$ to \CDO{} via a randomized reduction, thereby obtaining a CLB under the assumption that the Strong-\APSP{} hypothesis holds even for  randomized algorithms (either in expectation or with high probability).

\begin{theorem}\label{thm:CDO LB rand}
    Assuming the Strong-\APSP{} hypothesis for randomized algorithms, any algorithm for \CDO{} on $n$ points in an array of size $O(n)$ with $\Tilde{O}(n^a)$ preprocessing time and $\Tilde{O}(n^b)$ query time, respectively, must obey $a+b \ge 2$.
\end{theorem}

The proof of Theorem~\ref{thm:CDO LB rand}, which appears in \Cref{sec:CDO lb rand}, is based on the proof of Theorem~\ref{thm:minplus LB}, combined with probabilistic techniques in order to create instances that are not multi-colored.

Theorem~\ref{thm:CDO LB rand} combined with Theorem~\ref{thm:ACDO} implies that the exact version of \CDO{} is strictly harder than the approximate version.
Moreover, we remark that the CLBs for \CDO{} apply also to the snippets problem, since the special case of $\mathbf M$ defined by locations in an array is captured by the snippets problem when $P_1$ and $P_2$ are each a single character.
Thus, the CLB of Theorem~\ref{thm:CDO LB rand} applies also to the snippets problem, and so the exact version of the snippets problem is also strictly harder than the approximate version.

\section*{Additional Related Work}

A problem closely related to the \CDO{}  is the (approximate) vertex-labeled distance oracles for graphs (VLDO) problem, where the goal is to preprocess a colored graph \( G \) with $n$ vertices, so that given a vertex query \( v \) and a color \( c \), one can quickly return (an approximation of) \( d(v, c) \).
Hermelin, Levy, Weimann, and Yuster~\cite{HLWY11} introduced the VLDO problem and designed a solution using \( O(kn^{1 + 1/k}) \) expected space, with stretch factor  \( 4k - 5 \) and \( O(k) \) query time.
They also showed how to reduce the space usage to \( O(kN^{1/k}) \) but the stretch factor is exponential in \( 2k - 21 \).
Chechik~\cite{che12} later showed how to lower the stretch back to \( 4k - 5 \).
For planar graphs,
Evald, Fredslund-Hansen, and Wulff-Nilsen~\cite{EFW21} designed a near-optimal exact tradeoff using $n^{1+o(1)}$ space with $\Tilde{O}(1)$ query time, or  $\Tilde{O}(n)$ space with $n^{o(1)}$ query time.
Li, Ma, and Ning~\cite{LMN11} designed a $1+\varepsilon$ approximation  for planar graphs.
\section{Preliminaries and Algorithmic Overview}
Let $[n] = \{1,2,\ldots,n\}$.
Let $S \subset  \mathbb{Z}$ be a set of  $n$ integers.
Let $\max(S)$ be the largest integer in $S$ and let $\min (S)$ be the smallest integer in $S$.
For integers $i,j\in [n]$ where $i< j$, let $S[i,j] = \{p \in S : i \leq p \leq j\}$.

Our algorithms make use of the following Nearest Neighbor Search (NNS) data structures.

\begin{problem}
    [The Nearest Neighbour Search problem ~\cite{AGN14,An09,Kl97,Van75}] Given a set $S\subseteq \mathbf M$ of size $n$, preprocess $S$ to support the following query: given an integer $p$,
    return $\argmin_{p' \in S}\{d(p,p')\}$.
\end{problem}

\begin{problem}
    [The Range Nearest Neighbour Search (RNNS) problem ~\cite{BGKS14,BPS16,CIT12,KKFL14,KKL07,KL04,MNV16,NN12}]
    Given a set $S$ of $n$ integers, preprocess $S$ to support the following query: given $i,j\in [n]$ and an integer $p$,
    return $\argmin_{p' \in S[i,j]}\{d(p,p')\}$.
\end{problem}

\subsection{Algorithmic Overview}\label{sec:overview}

\paragraph*{Generic (Approximate) \Cdo{} Algorithm}
Our algorithm for \ACDO{} is based on a generic algorithm whose structure follows the structure of the algorithm  in~\cite{KK16} for the exact \Cdo{} problem.
The generic algorithm is a classic heavy-light algorithm; our algorithms are  a new  \emph{implementation} of the generic algorithm.

For an integer parameter $0 \leq \tau \leq n$, a color $c\in \mathcal{C}$ is said to be \emph{heavy} if $\abs{P(c)} \geq \tau$ and \emph{light} otherwise.
Let $\mathcal{H} = \{h_1,h_2,\ldots,h_{\abs{\mathcal{H}}}\}$ be the set of heavy colors, and let $\mathcal{L}$ be the set of light colors.
Notice that $\abs{\mathcal H}\leq \frac{n}{\tau}$.
In the preprocessing phase, for each color $c\in \mathcal{C}$, the algorithm stores $P(c)$ in an \NNS{}  data structure, denoted by  $\text{\NNS}_{c}$.
In addition, the algorithm pre-computes a  matrix $E^*=\{e^{*}_{i,j}\}$ of size $\abs{\mathcal{H}} \times \abs{\mathcal{H}}$, such that $\delta_{h_i,h_j} \le e^{*}_{i,j} \le (1+\varepsilon) \delta_{h_i,h_j}$.
An \ACDO{} query on $c,c' \in \mathcal{C}$ is processed as follows.
If both $c,c' \in \mathcal{H}$, then, without loss of generality, $c= h_i$ and  $c' = h_j$. In such a case the
algorithm returns $e^{*}_{i,j}$,
which is a $(1+\varepsilon)$ approximation of $\delta_{h_i,h_j}$.
Otherwise, without loss of generality, $C$ is a light color, and the algorithm returns $\min_{\hat{p}\in P(c)}\{d(\hat{p},c')\} = \delta_{c,c'}$, by performing   $\abs{P(c)}\le \tau$ \NNS{} queries on $\text{\NNS{}}_{c'}$, one for each point in $C$.

\paragraph*{Time complexity.}
The preprocessing and query time  costs of the generic algorithm depend on the implementation of the \NNS{} data structure, and the time used for computing matrix $E^*$.
Thus, we express the time complexity of the generic algorithm as a function of $T_{p,\text{\NNS}}(t)$, $T_{q,\text{\NNS}}(t)$, and  $T_{E^*}(\mathcal{H})$, which are the  preprocessing time and query time of the \NNS{} data structure on a set of size $t$, and the time used to compute $E^*$, respectively.
In the preprocessing phase, the algorithm computes $E^*$ and creates an $\text{\NNS{}}$ data structure for each color in $\mathcal C$.
Thus, the preprocessing time cost is $$O(T_{E^*}(\mathcal{H})+\sum_{c\in \mathcal C} T_{p,\text{\NNS{}}}(P(c))).$$
In the query phase, the time cost is dominated by executing at most $\tau$ \NNS{} queries on a set of size at most $n$, and so the time cost is
$O(\tau \cdot T_{q,{\text{\NNS}}}(n)).$

\paragraph*{Constructing $E^*$ and solving \ACDO{} on an array.}
In Section~\ref{sec:construct E*} we describe an efficient algorithm for constructing $E^*$ when $\mathbf{M}$ is the metric of integers.
Our algorithm is structured as follows.
For pairs of heavy colors whose color distance is smaller than some threshold, the algorithm computes their distance via a straightforward brute-force computation.
The more challenging part is dealing with pairs of heavy colors with larger color distance.
Our approach is to compute a series of Boolean matrices with specially designed properties, so that after the computation of the matrices we are able scan the matrices and deduce a $(1+\varepsilon)$ approximation for every pair of colors with a rather large distance color.
We remark that to compute the matrices used for the approximations, our algorithm makes use of FMM, which, as noted in Section~\ref{sec:intro}, is required for obtaining improvements in runtime compared to the exact solution of~\cite{KK16}.

To complete the proof of Theorem~\ref{thm:ACDO}, in Section~\ref{sec:1D-ACDO} we analyze the runtime of the generic algorithm with the runtime of constructing $E^*$ when $\mathbf{M}$ is defined by locations in an array, and by plugging in known \NNS{} data structures.

\paragraph*{Solving \AMCDOCH{} on an array.}
In Section~\ref{sec: HACDO by ACDO} we prove Theorem~\ref{thm: HACDO by ACDO}.
Our algorithm is a reduction from \AMCDOCH{} on an array to \ACDO{} on an array.
The structure of our reduction somewhat resembles the structure of the \MCDOCH{} algorithm of~\cite{KK16}.
In particular, in the algorithm of~\cite{KK16} they partition the points into sets of size $\tau$, preprocess each set into an \RNNS{} data structure, and for each pair of sets apply $\tau$ \RNNS{} queries to compute the distance between the sets.
To obtain a faster runtime, we compute the distance between every pair of sets by recoloring the input points using the sets as a new color scheme and then applying our \ACDO{} algorithm with the new color set.
One important property of our \ACDO{} algorithm is that when a query takes place between two heavy colors, the query time is constant since the algorithm just looks up a single value in $E^*$.
It turns out that essentially all of the colors in the new color set are heavy, and so after the preprocessing phase of our \ACDO{} algorithm, we are able to compute the distance between a pair of sets in $O(1)$ time.
This property turns out to be helpful for our \AMCDOCH{} algorithm, which is described in Section~\ref{sec: HACDO by ACDO}.

\section{Computing $E^*$ when $\mathbf{M}$ is integers}\label{sec:construct E*} In this section, we show how to compute matrix $E^*$ when our metric $\mathbf{M}$ is one dimensional, and for  $p,p'\in \mathbf M$ we have $d(p,p')=\abs{p-p'}$.
We make the simplifying assumption that $\min(S) = 1$ since,  otherwise, one could shift $S$ by subtracting $\min(S)-1$ from each point.

Let $W$ be a positive integer parameter that will be determined later.
The construction algorithm for $E^*$ has two conceptual parts.
The first part deals with pairs of heavy colors with distance at most $\frac{1+2\varepsilon}{\varepsilon}W$, by computing the distances exactly using a brute-force method.
The second part deals with pairs of heavy colors with distance greater than $\frac{1+2\varepsilon}{\varepsilon}W$.
In this case, the algorithm computes an approximation of the distances by utilizing a series of matrices defined as follows.

Let $\ell_0 = \floor{\log_{1+\varepsilon}(\frac{1}{\varepsilon})} + 1 $ and let $\ell_{max} = \ceil{\log_{1+\varepsilon}({\max{(S)}})}$.
Let $A=\{a_{i,j}\}$ be a Boolean $\abs{\mathcal{H}} \times \ceil{\frac{\max{(S)}}{W}}$ matrix where
  \[  a_{i,j} =
        \begin{dcases}
            1  &  P(h_i)\cap S[(j-1)W +1, jW] \neq \emptyset  \\
            0 & \text{otherwise}
        \end{dcases}
        \]
Thus, matrix $A$ indicates for each heavy color and each block of size $W$ in $S$ which starts at a location following an integer multiple of $W$, whether the heavy color appears in the block or not.

For $\ell_0\leq \ell \leq \ell_{max}$, let $B^{(\ell)}=\{b^{(\ell)}_{i,j}\}$ be a Boolean $\ceil{\frac{\max{(S)}}{W}}\times \abs{\mathcal{H}}$ matrix where
    \[b^{(\ell)}_{i,j}=
        \begin{dcases}
            1 & P(h_{j})\cap S[(i-1)W + 1, iW + (1+\varepsilon)^{\ell}W] \neq \emptyset \\
            0 & \text{otherwise
}        \end{dcases}
        \]
Matrix $B^{(\ell)}$ indicates for each block of size $W+(1+\varepsilon)^\ell W$ in $S$ which starts at a location following an integer multiple of $W$, and for each heavy color,  whether the heavy color appears in the block or not.

Let $E^{(\ell)} = A \cdot B^{(\ell)} = \{e^{(\ell)}_{i,j}\}$, where the product is a Boolean matrix product.
Thus, $E^{(\ell)}$ roughly indicates for each pair of heavy colors whether their color distance is at most $W+(1+\varepsilon)^\ell W$.

In  Section~\ref{sec:E ell matrix properties} we state and prove lemmas regarding the connections between the matrices $E^{(\ell)}$ and approximating color distances.
In Section~\ref{sec:E* contruct algo} we describe the algorithm for constructing $E^*$, and prove its correctness based on the lemmas of Section~\ref{sec:E ell matrix properties}.

\subsection{Properties of The $E^{(\ell)}$ Matrices}\label{sec:E ell matrix properties}

The following lemma describes a connection between entries of $0$ in $E^{\ell}$ and lower bounds on color distances.

\begin{lemma} \label{lem:lower-bound}
    For $h_i ,h_j \in \mathcal{H}$ and integer $\ell_0 \leq \ell \leq \ell_{max}$,
    if $e^{(\ell)}_{i,j} = 0$ and $e^{(\ell)}_{j,i} = 0$, then $\delta_{h_i,h_j} > (1+\varepsilon)^{\ell} W$.
\end{lemma}
\begin{proof}
    Let $p \in P(h_i)$ and $p' \in P(h_j)$ be chosen such that $\delta_{h_i,h_j} = d(p,p')$.
    Assume towards a contradiction that $\delta_{h_i,h_j} \leq (1+\varepsilon)^{\ell}W$.
    We focus on the case of $p' \geq p$ and so $\delta_{h_i,h_j} = p' - p$; the proof for the case $p > p'$ is symmetrical.
    Thus, there  exists  an integer $1\le k \leq \ceil{\frac{\max(S)}{W}}$ such that $p \in S[(k-1)W+1,kW]$ and so $a_{i,k} = 1$. Therefore,
    \begin{align*}
    p' = p + \delta_{h_i,h_j}
    \leq p + (1+\varepsilon)^{\ell}W \leq kW+  (1+\varepsilon)^{\ell}W.
    \end{align*}
     We conclude that $p' \in S[(k-1)W + 1, kW + (1+\varepsilon)^{\ell}W]$, implying that $b^{(\ell)}_{k,j} = 1$, and so $e^{(\ell)}_{i,j} =1$, which is a contradiction.
\end{proof}

The following lemma shows a connection between colors with distance greater than $(\frac{1+2\varepsilon}{\varepsilon})W$ and entries in $E^{(\ell_0)}$ whose value is 0.

\begin{lemma}\label{lem:lower-bound-base_case}
For $h_i,h_j \in \mathcal{H}$,
if $\delta_{h_i,h_j} > (\frac{1+2\varepsilon}{\varepsilon})W$ then $ e^{(\ell_{0})}_{i,j} = 0$ and  $e^{(\ell_{0})}_{j,i} = 0$.
\end{lemma}
\begin{proof}
Assume by contradiction that $e^{(\ell_{0})}_{i,j} = 1$ or $e^{(\ell_{0})}_{j,i} = 1$.
We focus on the case of $e^{(\ell_{0})}_{i,j} = 1$; the proof for the case $e^{(\ell_{0})}_{j,i} = 1$ is symmetrical.
Since $e^{(\ell_{0})}_{i,j} = 1$ , there exists an integer $1\le k \leq \ceil{\frac{\max(S)}{W}}$ such that $a_{i,k} = 1$ and $b^{(\ell_{0})}_{k,j} = 1$.
Therefore, there exist points $p \in P(h_i)$ and $p' \in P(h_j)$ such that $p \in S[(k-1)W + 1, kW] $ and $p' \in S[(k-1)W +1,kW + (1+\varepsilon)^{(\ell_0)}]$.
Thus, \begin{align*}
    \delta_{h_i,h_j} &= \min_{\hat{p}\in P(c), \Tilde{p}\in P(c')}\{d(\hat{p},\Tilde{p})\}
    \leq d(p,p')
    =\max\left(p - p', p' - p \right) \\
    &=\max\left(W - 1, (1+\varepsilon)^{(\ell_{0})} W+W -1  \right)
    = (1+\varepsilon)^{(\ell_{0})}W + W - 1\\
    &= (1+\varepsilon)^{\floor{\log_{1+\varepsilon}(\frac{1}{\varepsilon})} + 1}W + W - 1 \leq ((1+\varepsilon)^{\log_{1+\varepsilon}(\frac{1}{\varepsilon}) + 1} + 1)W  -1 \\
    &< (\frac{1+2\varepsilon}{\varepsilon})W.
\end{align*}
\end{proof}

The following lemma describes a connection between entries of $1$ in $E^{\ell}$ and upper bounds on color distances.

\begin{lemma} \label{lem:upper-bound}
    For $h_i,h_j \in \mathcal{H}$ and integer $\ell_{0} \leq \ell \leq \ell_{\max}$, if $e^{(\ell)}_{i,j} = 1$, then $\delta_{h_i,h_j} < (1+\varepsilon)^{\ell+1} W$.
\end{lemma}
\begin{proof}
If $e^{(\ell)}_{i,j} = 1$, then there  exists an integer  $1\le k \leq \ceil{\frac{\max(S)}{W}}$  such that $a_{i,k} = 1$ and $b^{(\ell)}_{k,j} = 1$.
Hence, there exist $p \in P(h_i)$ and $p' \in P(h_j)$ such that
$p \in S[(k-1)W + 1,kW]$ and $p' \in S[(k-1)W + 1,kW+ (1+\varepsilon)^{\ell}W]$.
Thus,
\begin{align*}
\delta_{h_i,h_j} &\le \abs{p - p'}
= \max(p - p', p' - p)
 \leq  \max\left(W - 1, (1+\varepsilon)^{\ell}W + W - 1\right)\\
 &= (1+\varepsilon)^{\ell}W + W - 1
 < ((1+\varepsilon)^{\ell}+1)W.
\end{align*}
Notice that $\floor{\log_{1+\varepsilon}(\frac{1}{\varepsilon})} +1 > \log_{1+\varepsilon}(\frac{1}{\varepsilon})$.
Finally, since $\floor{\log_{1+\varepsilon}(\frac{1}{\varepsilon})} + 1 = \ell_0\leq \ell$, we have
\begin{align*}
(1+\varepsilon)^{\ell} &\geq (1+\varepsilon)^{\floor{\log_{1+\varepsilon}(\frac{1}{\varepsilon})}+1}\
 > (1+\varepsilon)^{\log_{1+\varepsilon}(\frac{1}{\varepsilon})}
 = \frac{1}{\varepsilon},
\end{align*}
and so $1 < \varepsilon (1+\varepsilon)^{\ell}$.
Thus,
\begin{align*}
    ((1+\varepsilon)^{\ell}+1)W  < ((1+\varepsilon)^{\ell}+\varepsilon (1+\varepsilon)^{\ell})W
    = ((1+\varepsilon)^{\ell+1})W,
\end{align*}
completing the proof.
\end{proof}

The following lemma shows that $e^{(\ell)}_{i,j}$ is weakly monotone increasing with respect to $\ell$.

\begin{lemma} \label{lem:consistency}
  For $h_i,h_j \in \mathcal{H}$, and $\ell_0 \leq \ell \leq \ell_{max}$, if $e^{(\ell)}_{i,j} = 1$, then for any $\ell <\ell' \leq \ell_{max}$ we have $e^{(\ell')}_{i,j} = 1$.
\end{lemma}

\begin{proof}
If $e^{(\ell)}_{i,j} = 1$, then there exists an integer $1\leq k \leq \ceil{\frac{\max(S)}{W}}$ such that $a_{i,k} = 1$ and $b^{(\ell)}_{k,j} = 1$.
Thus, there exists $p \in P(h_i)$ such that
$p \in S[(k-1)W+1,kW + (1+\varepsilon)^{\ell} W]$.
Moreover, since $\ell < \ell'$, we have  $S[(k-1)W+1,kW+ (1+\varepsilon)^{\ell}W] \subseteq  S[(k-1)W+1,kW+ (1+\varepsilon)^{\ell'} W]$, and so $b^{(\ell')}_{k,j} =1$.
Finally, since $a_{i,k} = 1$, we conclude that $e^{(\ell')}_{i,j} = 1$.
\end{proof}

The following lemma shows that in $E^{(\ell_{max})}$, for each pair of heavy colors, there exists at least one corresponding 1 entry.

\begin{lemma} \label{lem:final_marix_M}
 For $h_i,h_j \in \mathcal{H}$, either $e_{i,j}^{(\ell_{max})} = 1$ or $e_{j,i}^{(\ell_{max})} = 1$.
\end{lemma}

\begin{proof}
Assume by contradiction that $e_{i,j}^{(\ell_{max})} = 0$ and $e_{j,i}^{(\ell_{max})} = 0$.
Let $p \in P(h_i)$ and $p' \in P(h_j)$ be chosen such that $\delta_{h_i,h_j} = d(p,p')$.
Hence by Lemma~\ref{lem:lower-bound}, we have $\delta_{h_i,h_j} >  (1+\varepsilon)^{\ell_{max}}W = (1+\varepsilon)^{\ceil{\log_{1+\varepsilon}{\frac{\max(S)}{W}}}}W$.
Notice that $\max{(S)} \geq \delta_{h_i,h_j}$.
Therefore,
\begin{align*}
 \max{(S)}  &\geq \delta_{h_i,h_j}
 >(1+\varepsilon)^{(\ceil{\log_{1+\varepsilon}{\frac{\max{(S)} }{W}}})}W \geq (1+\varepsilon)^{\log_{1+\varepsilon}{\frac{\max{(S)} }{W}}}W
\geq  \max{(S)},
\end{align*}
which is a contradiction.
\end{proof}

In the following lemma we show the main property of the matrices $E^{(\ell)}$ used in the algorithm of Section~\ref{sec:E* contruct algo} when dealing with pairs of heavy colors with color distance greater than $\frac{1 + 2\varepsilon}{\varepsilon}W$.
\begin{lemma} \label{lem: uniqueness of a lower bound and upper bound}
    For $h_i,h_j \in \mathcal{H}$ , if $\delta_{h_i,h_j} > \frac{1+2\varepsilon}{\varepsilon}W$,  then there exists a unique integer $\ell_{0} \leq \ell \leq \ell_{max} - 1$ such that $e_{i,j}^{(\ell)} = e_{j,i}^{(\ell)} = 0$ and either  $e_{i,j}^{(\ell+1)} = 1$
    or $e_{j,i}^{(\ell+1)} = 1$.
\end{lemma}
\begin{proof}
    By Lemma~\ref{lem:final_marix_M}, at least one of $e_{i,j}^{(\ell_{max})}$ and $ e_{j,i}^{(\ell_{max})}$ must  be $1$.
    So assume without loss of generality that $e_{i,j}^{(\ell_{max} )} = 1$.
    Since $\delta_{h_i,h_j} > \frac{1+2\varepsilon}{\varepsilon}W$, by Lemma~\ref{lem:lower-bound-base_case} $e_{i,j}^{(\ell_{0} )} = 0$.
    Thus, by Lemma~\ref{lem:consistency}, there must exist exactly one  $\ell_0\le \ell < \ell_{max}$ such that $e_{i,j}^{(\ell)} = 0$ and $e_{i,j}^{(\ell+1)} = 1$.

    To complete the proof we show uniqueness.
    If $e_{i,j}^{(\ell_{max})} = e_{j,i}^{(\ell_{max})} = 1$, then
    assume towards a contradiction that there exist two different integers $\ell'\neq \hat \ell$  such that
     $e_{i,j}^{(\ell')} = 0$, $e_{j,i}^{(\ell')} = 0$,   $e_{i,j}^{(\ell'+1)} = 1$,   $e_{i,j}^{(\hat{\ell})} = 0$,  $e_{j,i}^{(\hat{\ell})} = 0$, and  $e_{j,i}^{(\hat{\ell}+1)} = 1$.
    Without loss of generality,  $\ell' < \hat{\ell} $, and so  by Lemma~\ref{lem:consistency}, $e_{i,j}^{(\ell'+1)} = 1$ and $\ell'<\hat \ell$ implies $e^{(\hat{\ell} )}_{i,j} = 1$ which is a contradiction.

Otherwise, if at least one of $e_{i,j}^{(\ell_{max})}$ or $ e_{j,i}^{(\ell_{max})}$ is $0$, then since we assumed without loss of generality that $e_{i,j}^{(\ell_{max} )} = 1$, we have $e_{j,i}^{(\ell_{max})}=0$, which together with Lemma~\ref{lem:consistency} implies that for all  $\ell_0\le \ell' < \ell_{max}$,  $e_{j,i}^{(\ell')}=0$.
Thus, $\ell$ is unique.
\end{proof}

\subsection{Algorithm for Computing $E^*$.}\label{sec:E* contruct algo}

\begin{algorithm}
\caption{\textsc{Construct}$E^*(S,\mathcal H,W,\varepsilon)$}
\label{alg:building-E-star}
\begin{algorithmic}[1]
\State $E^* \leftarrow \infty$
\State $\ell_0 \leftarrow \floor{\log_{1+\varepsilon}(\frac{1}{\varepsilon})} + 1 $;  $\ell_{max} \leftarrow \ceil{\log_{1+\varepsilon}({\max{(S)}})}$
\For {$p \in S$ and $\hat{p} \in S[p, p + \frac{1+2\varepsilon}{\varepsilon}W]$} \label{line:start_compute_small_distance}
        \If{$C_p = h_i$ and $C_{\hat{p}} = h_i$ and $e^{*}_{i,j} > \hat{p} - p$}
            \State $e^*_{i,j} \leftarrow \hat{p} - p$; $e^*_{j,i} \leftarrow \hat{p} - p$
        \EndIf
\EndFor\label{line:end_compute_small_distance}
\State Construct matrix $A$\label{line:begin_matrix_comp}
\For{$\ell \leftarrow \ell_{0} $ to $\ell_{max}$}
\State Construct matrix $B^{(\ell)}$
    \State $E^{(\ell)} \leftarrow A \cdot B^{(\ell)}$ \label{Mult}
\EndFor    \label{line:end_matrix_comp}
\For{$\ell \leftarrow \ell_{0} $ to $\ell_{max}$}\label{line:begin_update}
    \For{$(h_i,h_j) \in \mathcal H \times \mathcal H$}
     \If{$e^{(\ell)}_{i,j} = e^{(\ell)}_{j,i}= 0$ and  ($e^{(\ell+1)}_{i,j} = 1$ or $e^{(\ell+1)}_{j,i} = 1$)}
            \State $e^*_{i,j} \leftarrow (1+\varepsilon)^{(\ell+2)}W$\label{line:assign eij}
            \EndIf
    \EndFor\label{line:end_update}
\EndFor

\end{algorithmic}
\end{algorithm}

In this section, we describe the algorithm for computing $E^*$, prove its correctness and analyze its time cost.
Pseudocode for the algorithm is given in Algorithm 1.
The algorithm begins by a brute-force exact computation of color distances for all heavy colors $h_i$ and $h_i$ where $\delta_{h_i,h_j} \le \frac{1+2\varepsilon}{\varepsilon}W$.
The brute-force computation costs $O((\frac{1+2\varepsilon}{\varepsilon})Wn) = O(\frac{Wn}{\varepsilon})$ time.

The rest of the algorithm deals with the case of $\delta_{h_i,h_j} > \frac{1+2\varepsilon}{\varepsilon}W$.
First, the algorithm computes the matrices $E^{(\ell)}$ from matrices $A$ and $B^{(\ell)}$.
The time cost of computing all $E^{(\ell)}$ is dominated by the cost of $\log_{1+\varepsilon}(\max{(S)}) =\Theta(\frac{ \log \max(S)} \varepsilon)$ matrix multiplications.
Each multiplication is between a matrix of size $\frac n \tau \times \frac {\max(S)} W$ and a matrix of size $ \frac {\max(S)} W \times \frac n \tau$.
It is folklore knowledge that the matrix multiplication of two matrices of size $x\times y$ and $y\times z$ can be computed in $O\left(\frac {x\cdot y \cdot z} {\min(x,y,z)^{3-\omega}}\right)$.
Thus, the time cost of computing the matrix multiplications is $O\left(
\frac{n^2\cdot \max{(S)}}{\tau^2W \cdot \min(\frac{n}{\tau}, \frac{\max{(S)}}{W})^{3 - \omega}}
\frac{ \log \max(S)} \varepsilon \right)$ time.

Finally, the algorithm utilizes the matrices $E^{(\ell)}$  to complete  entries in $E^*$ which correspond to heavy pairs of colors that were not covered by the brute-force computation.
To do so, for each such pair $(h_i,h_j) \in \mathcal H \times \mathcal H$, the algorithm scans all $e_{i,j}^{(\ell)}$ to find the unique $\ell$ from Lemma~\ref{lem: uniqueness of a lower bound and upper bound} such that $e^{(\ell)}_{i,j} = e^{(\ell)}_{j,i}= 0$ and  either $e^{(\ell+1)}_{i,j} = 1$ or $e^{(\ell+1)}_{j,i} = 1$, and  sets $e^*_{i,j}$  to be $(1+\varepsilon)^{\ell+2}W$.
Computing the entries in $E^*$ for all such pairs costs $O(|\mathcal H|^2\frac{ \log \max(s)} \varepsilon)) = O((\frac n \tau)^2\frac{ \log \max(s)} \varepsilon))$, which is dominated by the cost of the matrix multiplications performed during the computation of all of the $E^{(\ell)}$ matrices.

\begin{lemma}\label{lem:main lemma}
There exists an algorithm that computes $E^*$ in $ \Tilde{O}\left(\frac{Wn}{\varepsilon} +\frac{n^2\cdot \max{(S)}}{\tau^2W \cdot \min(\frac{n}{\tau}, \frac{\max{(S)}}{W})^{3 - \omega}}\right)$ time, such that  for $h_i ,h_j \in \mathcal{H}$ we have
      $\delta_{h_i,h_j} \leq e^*_{i,j} \leq   (1+\varepsilon)\delta_{h_i,h_j}$.

\end{lemma}
\begin{proof}
The runtime of the algorithm follows from the discussion above.

If $\delta_{h_i,h_j} \leq \frac{1+2\varepsilon}{\varepsilon}W$, then there exist $p\in P(h_1)$ and $\hat p\in P(h_i)$ such that $\delta_{h_i,h_j} = \abs{\hat p-p} \leq \frac{1+2\varepsilon}{\varepsilon}W$,
and so, after the brute-force computation, we have $e^*_{i,j}= \delta_{h_i,h_j}$.

Otherwise,  $\delta_{h_i,h_j} > \frac{1+2\varepsilon}{\varepsilon}W$, and so the algorithm sets $e^*_{i,j}$  to be $(1+\varepsilon)^{\ell+2}W$, where $\ell$ is the unique integer from Lemma~\ref{lem: uniqueness of a lower bound and upper bound}, and in particular, $e_{i,j}^{(\ell)} = e_{j,i}^{(\ell)} = 0$ and either  $e_{i,j}^{(\ell+1)} = 1$
    or $e_{j,i}^{(\ell+1)} = 1$.
By Lemmas~\ref{lem:upper-bound} and~\ref{lem:lower-bound}, $
\delta_{h_i,h_j} < (1+\varepsilon)^{\ell + 2}W< (1+\varepsilon)^2\delta_{h_i,h_j} $.

Finally, to obtain a $(1+\varepsilon)$ approximation one can run the algorithm with approximation parameter $\varepsilon' = \frac{\varepsilon}{3}$, which does not affect the asymptotic time complexity.
\end{proof}

\section{Proof of Theorem~\ref{thm:ACDO}}\label{sec:1D-ACDO}
In this section, we analyze the combination of the construction algorithm for $E^*$ given in Section~\ref{sec:E* contruct algo} together with the generic \ACDO{} algorithm described in Section~\ref{sec:overview}.
However, we focus on the type of instances of $\ACDO{}$ which are used in Section~\ref{sec: HACDO by ACDO} for solving the approximate snippets problem.
In particular, our metric is defined by locations in an array of size $n$,  and so we have $\max(S)=n$.
Thus, in our case,
\begin{align*}
   T_{E^*}(\mathcal{H}) &= O\left(\frac{Wn}{\varepsilon} +\frac{n^{\omega}\max(\tau,W)^{3-\omega}}{\tau^2W }\cdot\frac{ \log n} \varepsilon \right).
\end{align*}
If $n^\frac{\omega -1}{\omega + 1}\log^\frac{1}{2} n\leq \tau\leq n$, then we choose $W =  (\frac{n} {\tau})^{\frac{\omega -1 }{2}}\log^\frac{1}{2} n$.
In such a case, we have
$W = (\frac{n} {\tau})^{\frac{\omega -1 }{2}} \log^{\frac{1}{2}}n
    \le \left(\frac{n}{n^\frac{\omega -1}{\omega + 1}}\right)^{\frac{\omega -1 }{2}}\log^\frac{1}{2} n
     = n^{\frac{\omega -1}{\omega + 1}}\log^\frac{1}{2} n
    \leq \tau$, and so
\begin{align*}
    T_{E^*}(\mathcal{H})&=O\left(\frac{(\frac{n}{\tau})^{\frac{\omega -1 }{2}}n}{\varepsilon}\log^{\frac{1}{2}}n +\frac{n^{\omega}\tau^{3-\omega}}{\tau^2(\frac{n}{\tau})^{\frac{\omega -1 }{2}} \cdot \log^{\frac{1}{2}}n}\cdot\frac{\log n} \varepsilon \right)
    =O\left(\frac{n^{\frac{\omega +1 }{2} }} {\varepsilon\tau^{\frac{\omega -1 }{2}}} \log^\frac{1}{2} n\right).
\end{align*}

If $1 \leq \tau \leq n^\frac{\omega -1}{\omega + 1}\log^\frac{1}{2} n$,  then we choose
$W = \frac{n}{\tau^{\frac{2}{\omega-1}}}\cdot \log^\frac{1}{\omega-1} n$.
In such a case we have
$W =
\frac{n}{\tau^{\frac{2}{\omega-1}}}\log^{\frac{1}{\omega - 1 }} n
\geq
\frac{n}{n^{\frac{\omega -1}{\omega +1}\cdot{\frac{2}{\omega - 1}}}}\log^\frac{1}{\omega-1} n
= n^\frac{\omega - 1 }{\omega+ 1 }\log ^\frac{1}{\omega-1} n
\geq \tau$,
and so
\begin{align*}
    T_{E^*}(\mathcal{H})&=O\left(\frac{n^2}{\varepsilon\tau^{\frac{2}{\omega-1}}}\log^\frac{1}{\omega-1} n +\frac{n^{\omega}(\frac{n}{\tau^\frac{2}{\omega - 1}})^{2-\omega} \cdot \log(n)^{\frac{2 - \omega}{\omega - 1 }}}{\tau^2} \cdot\frac{ \log n} \varepsilon \right)
    =O\left(\frac{n^{2}}{\varepsilon\tau^{\frac{2}{\omega -1 }}} \log^\frac{1}{\omega-1}n \right).
\end{align*}
Thus, to summarize we have shown
\[T_{E^*}(\mathcal H) = \begin{dcases}
      \Tilde{O}\left(\frac{n^{2}}{\varepsilon\tau^{\frac{2}{\omega -1 }}}\right) &\text{for } 1 \leq  \tau \leq n^\frac{\omega -1}{\omega + 1}\\
      \Tilde{O}\left(\frac{n^{\frac{\omega +1 }{2} }}{\varepsilon\tau^{\frac{\omega -1 }{2}}}\right) &\text{for } n^\frac{\omega -1}{\omega + 1} \leq \tau \leq n.
    \end{dcases}\]
Notice that in either case, $T_{E*}(\mathcal{H}) = \Omega(\max(n^2/\tau^2),n)$.

For the \NNS{} data structure we use the van Emde Boas~\cite{Van75}  data structure which  for a set of $m$ points from integer universe $\{1,2,\ldots,u\}$ has a preprocessing cost of $O({m}\cdot \log\log{u})$ and query time $O(\log\log{u})$.
In our setting, $u=n$, and so $T_{p,\NNS}(P(c)) = O(|P(c)| \log \log n)$ and $T_{q,\NNS}(n) = O(\log \log n)$.

Thus, the construction time of the \ACDO{} algorithm is
\begin{align*}
    O(T_{E^*}(\mathcal{H})+\sum_{c\in \mathcal C} T_{p,\text{\NNS{}}}(P(c))) & =     O(T_{E^*}(\mathcal{H})+\sum_{c\in \mathcal C}  |P(c)| \log \log n)\\
    &=O(T_{E^*}(\mathcal{H})+n \log \log n)\\
    & = \tilde O(T_{E^*}),
\end{align*}
and the query time is $O(\tau \cdot T_{q,{\text{\NNS}}}(n)) = O(\tau \log \log n) = \tilde O(\tau)$.
To complete the proof of Theorem~\ref{thm:ACDO}, we plug $\tau = \tilde O(n^b)$ into $n^a = T_{E^*}(\mathcal{H})$, to obtain the following tradeoff curve for a fixed $\varepsilon$.
\[\begin{dcases}
        a +  \frac{2}{\omega - 1}b = 2 &\text{if } 0 \leq  b \leq \frac{\omega -1}{\omega + 1} \\
        \frac{2}{\omega - 1}a + b = {\frac{\omega+1}{\omega-1}}  &\text{if } \frac{\omega -1}{\omega + 1} \leq b \leq 1.
    \end{dcases}\]

\begin{remark}\label{rmk:constant time queries}
One feature of  Theorem~\ref{thm:ACDO} algorithm, which is used in the proof of Theorem~\ref{thm: HACDO by ACDO}, is that the query cost for \ACDO{} queries when both $C$ and $C'$ are heavy is $O(1)$ time since all the algorithm does is looking up the answer in $E^*$.
\end{remark}

\section{Proof of Theorem~\ref{thm: HACDO by ACDO}}\label{sec: HACDO by ACDO}
In this section we present an algorithm which solves the \AMCDOCH{} problem.
The algorithm is inspired by the exact \MCDOCH{} algorithm of~\cite{KK16}, combined with a new observation which allows to leverage Theorem~\ref{thm:ACDO}.

Assume without loss of generality that $T_S$ (the tree representing the color hierarchy) is an ordinal tree.
Moreover, by a reduction given in~\cite{KK16}, we may assume that there is a bijection embedding from the points in $S$ to leaves of $T_S$.
We abuse notation and say that a point $p\in S$ is \emph{at the leaf} corresponding to $p$ in $T_S$.
Let $A$ be an array containing the set of points in $S$, but with the order defined by the order in which points in $S$ are encountered during a pre-order traversal of $T_S$.
Constructing $A$ costs linear time in the size of $T_S$.
As observed in~\cite{KK16}, for each $c \in \mathcal C$, there exist integers $1\le x_c \le y_c \le n$ such that $P(c)= A[x_c,y_c]$.

The algorithm preprocesses $A$ into an \RNNS{}~\cite{KL04,KKL07,CIT12,NN12,BGKS14,KKFL14,BPS16,MNV16}
data structure.
Next, the algorithm partitions $A$ into blocks of size $\tau$, for an integer parameter $1\le \tau \le n$ to be determined later, as follows.
For integer $1\le i \le \ceil{\frac n \tau}-1$  let $I_i = A[(i-1)\tau+1, i\tau]$.
Let $I_{\ceil{\frac n \tau}} = A[(\frac n \tau-1)\tau+1),n]\}$.
Let $\mathcal{I} = \{I_1, I_2 \ldots I_{\ceil{\frac{n}{\tau}}}\}$.
Notice that each $I_i$ contains $\tau$ points from $S$ (except for possibly $I_{\ceil{\frac{n}{\tau}}}$ which contains at most $\tau$ points from $S$), but their order depends on their locations in $A$.
For a color $c\in \mathcal C$, if $y_c - x_c \ge 2\tau$ then the subarray $A[x_c,y_c]$ can be partitioned into three parts: (1) the prefix $J_c = A[x_c,\ceil{\frac{x_c -1 }{\tau}}\tau]$, (2) the middle part which is the union of blocks $I_c = \bigcup_{\ceil{\frac{x_c -1 }{\tau}} +1 \le t \le \floor{\frac{y_c }{\tau}}} I_t = A[\ceil{\frac{x_c -1 }{\tau}} \tau +1, \floor{\frac{y_c }{\tau}}\tau] $, and (3) the suffix $K_c = A[(\floor{\frac{y_c }{\tau}}\tau + 1,y_c]$.
Notice that the sizes of $J_c$ and $K_c$ are less than $\tau$, and $J_c\cup I_c\cup K_c = A[x_c,y_c]$.

Following~\cite{KK16}, the algorithm constructs a matrix $B=\{b_{i,j}\}$ of size $\ceil{\frac{n}{\tau}} \times \ceil{\frac{n}{\tau}}$, such that $b_{i,j} = d(I_i,I_j)$.
To construct $B$, the algorithm of~\cite{KK16} performs $\tau$ \RNNS{} queries for the computation of each entry in $B$.
Computing ${B}$ turns out to be the dominating component in the preprocessing runtime of~\cite{KK16}.
\paragraph*{Approximating matrix $B$.}
To obtain a faster preprocessing time for the approximate setting, instead of constructing $B$, our algorithm constructs an \emph{approximate} matrix $\hat B = \{\hat b _{i,j}\}$ where $d(I_i,I_j) \le \hat b_{i,j} \le (1+\varepsilon) d(I_i,I_j)$.
The construction of $\hat B$ utilizes the algorithm of Theorem~\ref{thm:ACDO} as follows.
Define a new coloring set $\hat{\mathcal{C}} = \{ \hat{c_1},\hat{c_2}, \ldots \hat{C_{\abs{\mathcal{I}}}} \}$ over the points in $S$,
such that for each $\hat{c_i} \in \hat{\mathcal{C}}$ we have $P(\hat{c_i}) = I_i$.
Notice that for every $i\ne j$, we have $I_i \bigcap I_j  = \emptyset{}$, and so using the colors of $\hat{\mathcal C}$ on $S$, every point in $S$ has only one color.
Thus, the algorithm uses the algorithm of Theorem~\ref{thm:ACDO} on $S$, but with the colors of $\hat{\mathcal C}$, and the query time designed to be $n^b = \tau$.
Now,  $\hat b_{i,j} $ is set to be the answer of the \ACDO{} query on colors $\hat c_i$ and $\hat c_j$, and so $d(I_i,I_j)  = \delta_{\hat c_i, \hat c_j} \le \hat b_{i,j} \le (1+\varepsilon) \delta_{\hat c_i, \hat c_j} = (1+\varepsilon) d(I_i,I_j)$.

In the last step of the preprocessing phase, the algorithm preprocesses $\hat B$ using a $2D$ Range
Minimum Query (\DRMQ{}) data structure~\cite{AFL07} so that given a rectangle in $\hat B$, defined by its corners, the algorithm  returns in $O(1)$ time the smallest value entry in the rectangle.

\paragraph*{Answering queries}
In order to answer an \AMCDOCH{} query between $C$ and $c'$ taken from color set $\mathcal C$, if either $P(c)\subseteq P(c')$ or $P(c')\subseteq P(c)$, which can be tested in $O(1)$ time using $x_c,y_c,x_{c'},$ and $y_{c'}$, then $\delta_{c,c'} = 0$ and the algorithm returns $0$.
Thus, for the rest of the query process we assume that $P(c)$ and $P(c')$ are disjoint.

If $y_c - x_c < 2\tau$, then the algorithm returns $
\min_{p\in A[x_c, y_c]}\{d(p,A[x_{c'},y_{c'}])\} = \delta_{c,c'}$ by applying $|P(c)|$ \RNNS{} queries.
Similarly, if $y_{c'} - x_{c'} < 2\tau$, then the algorithm returns $
\min_{p'\in A[x_{c'}, y_{c'}]}\{d(A[x_{c},y_{c}],p'])\} = \delta_{c,c'}$ by applying $|P(c')|$ \RNNS{} queries.

Otherwise, $y_c - x_c \ge 2\tau$ and $y_{c'} - x_{c'} \ge 2\tau$, and so the algorithm can use the partitions of $A[x_c,y_c] = J_c\cup I_c\cup K_c$ and  $A[x_{c'},y_{c'}] = J_{c'}\cup I_{c'}\cup K_{c'}$.
Notice that
\begin{align*}
    \delta_{c,c'} &= d(A[x_c,y_c],A[x_{c'},y_{c'}])\\
    &=\min \left( d(J_c\cup K_c, A[x_{c'},y_{c'}]),d(A[x_{c},y_{c}], J_{c'}\cup K_{c'}), d(I_c, I_{c'})  \right)
\end{align*}

To compute $d(J_c\cup K_c,A[x_{c'},y_{c'}])$, the algorithm executes an \RNNS{} query between each $p\in  J_c\cup K_c$ and $A[x_{c'},y_{c'}]$,  and sets  $\alpha_{c}(c') = \min_{p\in J_c\cup K_c}\{d(p,A[x_{c'},y_{c'}])\} = d(J_c\cup K_c,A[x_{c'},y_{c'}])$.
Similarly, the algorithm sets $\alpha_{c'}(c) = \min_{p\in J_{c'}\cup K_{c'}}\{d(p,A[x_{c},y_{c}])\} = d(A[x_{c},y_{c}],J_{c'}\cup K_{c'})$.
Thus, computing $\alpha_c(c')$ and $\alpha_{c'}(c)$ costs  $O(\tau)$ \RNNS{} queries.

Instead of computing $d(I_c,I_{c'})$ exactly,
the algorithm utilizes $\hat B$ to compute a $(1+\varepsilon)$ approximation of $d(I_c,I_{c'})$ as follows.
Since $P(c)$ and $P(c')$ are disjoint, assume without loss of generality that $x_c \le y_c < x_{c'} \le y_{c'}$.
The algorithm
executes a \DRMQ{} data structure query on the rectangle in $\hat B$ defined by corners $(\ceil{\frac{x_c-1}{\tau}} +1, \ceil{\frac{x_{c'}-1}{\tau}} +1)$ and $(\floor{\frac{y_c}{\tau}}, \floor{\frac{y_{c'}}{\tau}})$.
Let $\hat{b}_{i,j}$ be the answer returned by the \DRMQ{} query.
When proving correctness we will show that $\hat{b}_{i,j}$ is a $(1+\varepsilon)$ approximation of $d(I_c,I_{c'})$.
Finally, the algorithm returns $\min\left( \hat{b}_{i,j},\alpha_c(c'),\alpha_{c'}(c) \right).$

\paragraph*{Time Complexity.}
For the \RNNS{} data structure we use the solution of~\cite{NN12} which on $n$ points (which is the size of $A$) has  preprocessing time $O(n\log n)$ and query time  $O(\log^{\varepsilon} n ) = \Tilde{O}(1)$.

By Remark~\ref{rmk:constant time queries} and the fact that $|I_i| = |P(\hat C_i)| = |P(\hat C_j)| = |I_j| = \tau = n^b$, each \ACDO{} query used for constructing $\hat B$ costs $O(1)$ time\footnote{There is subtle technical issue here, since $I_{\ceil{\frac{n}{\tau}}}$ may contain less than $\tau$ elements. To solve this, we add dummy points with color $\hat C_{\ceil{\frac{n}{\tau}}}$ at locations $[2n+1,\dots 2n+\tau-|I_{\ceil{\frac{n}{\tau}}}|]$.
Thus, $\abs{P(\hat C_{\ceil{\frac{n}{\tau}}})}=\tau$, and the  dummy points do not affect the color distances (since they are far enough from the non-dummy points).}.
So after preprocessing  the \ACDO{} data structure, constructing  $\hat B$ costs $O((\frac{n}{\tau})^2)$ time.
The preprocessing time of the \DRMQ{} data structure of~\cite{AFL07} when applied to $\hat B$ is $\tilde O((\frac{n}{\tau})^2)$.
Thus, since $T_{E^*}(n) = \Omega(\max(n^2/\tau^2),n))$, the total time cost of the preprocessing phase, which is composed of constructing the \ACDO{} and \RNNS{}  data structures on $A$, computing $\hat B$, and preprocessing a \DRMQ{} data structure, is
$$ \Tilde{O}\left(T_{E^*}(n) + (n/{\tau})^2 + n \right) = \Tilde{O}\left(T_{E^*}(n)\right) =
\begin{dcases}
      \Tilde{O}\left(\frac{n^{2}}{\varepsilon\tau^{\frac{2}{\omega -1 }}}\right) &\text{for } 1 \leq  \tau \leq n^\frac{\omega -1}{\omega + 1}\\
      \Tilde{O}\left(\frac{n^{\frac{\omega +1 }{2} }}{\varepsilon\tau^{\frac{\omega -1 }{2}}}\right) &\text{for } n^\frac{\omega -1}{\omega + 1} \leq \tau \leq n
    \end{dcases} $$

The query process consists of $O(\tau)$ \RNNS{} queries, and a single \DRMQ{} query.
Thus, the query time cost is $\Tilde{O}(\tau)$.
Finally, similar to the proof of Theorem~\ref{thm:ACDO}, we  obtained the following tradeoff curve for a fixed $\varepsilon$.
\[\begin{dcases}
        a +  \frac{2}{\omega - 1}b = 2 &\text{if } 0 \leq  b \leq \frac{\omega -1}{\omega + 1} \\
        \frac{2}{\omega - 1}a + b = {\frac{\omega+1}{\omega-1}}  &\text{if } \frac{\omega -1}{\omega + 1} \leq b \leq 1.
    \end{dcases}\]

\paragraph*{Correctness.}
If either $y_{c} - x_{c} < 2\tau$ or  $y_{c'} - x_{c'} < 2\tau$ then the algorithm returns $\delta_{c,c'}$ exactly.
Thus, for the rest of the correctness proof we focus on the case where $y_{c} - x_{c} \geq 2\tau$ and $y_{c'} - x_{c'} \geq 2\tau$.
Notice that for every two sets $S$ and $S'$, and a subset $\hat S \subset S$ we have $d(S,S') \leq d(\hat S,S')$.
Furthermore, since for every integer $\ceil{\frac{x_{c} - 1}{\tau}} +1 \leq t \leq \floor{\frac{y_{c}}{\tau}}$ it holds that $I_{t} \subseteq I_{c}$, we have $I_{i} \subset A[x_c,y_c] = P(c)$ and $I_{j} \subset A[x_{c'},y_{c'}] = P({c'})$.
In addition, by Theorem~\ref{thm:ACDO} $\hat{b}_{i,j} \geq \delta_{\hat C_i, \hat C_j} = d(I_i,I_j)$.
Therefore,
$$
\hat{b}_{i,j} \geq d(I_i,I_j)
\geq d(A[x_{c},y_{c}],I_j)
\geq d(A[x_{c},y_{c}],A[x_{c'},y_{c'}])
=\delta_{c,c'}.$$

Furthermore, since  $\alpha_{c}(c') =d(J_c\cup K_c,A[x_{c'},y_{c'}])$ and $J_c \cup K_c \subset A[x_c,y_c]$, we have $\alpha_{c}(c') \geq \delta_{c,c'}$.
Similarly, $\alpha_{c'}(c) \geq \delta_{c,c'}$.
Thus, we conclude that $\delta_{c,c'} \leq \min\left( \hat{b}_{i,j},\alpha_c(c'),\alpha_{c'}(c) \right)$.

Next, notice that there exist $p \in A[x_c,y_c]$ and $p' \in A[x_{c'},y_{c'}]$ such that $\delta_{c,c'} = d(p,c')= d(c,p') = d(p,p')$.
Therefore, if $p\in J_{c}\cup K_{c}$, then $\alpha_c(c') = d(J_c\cup K_c,A[x_{c'},y_{c'}]) = d(p,p') = \delta_{c,c'}$, and so
$$    \delta_{c,c'} \leq \min\left( \hat{b}_{i,j},\alpha_c(c'),\alpha_{c'}(c) \right) \leq \alpha_c(c')
= \delta_{c,c'}.
$$
Thus, all of the inequalities become
equalities, and so, in this case, the algorithm returns $\delta_{c,c'}$.
Similarly, if $p'\in J_{c'}\cup K_{c'}$ then the algorithm returns $ \alpha_{c'}(c) = \delta_{c,c'}$.

For the remaining cases, assume that $p\not\in J_c \cup K_c $ and $p' \not\in J_{c'} \cup K_{c'}$.
Therefore, it must be that $p\in I_c$ and $p'\in I_{c'}$, and so there exist integers  $\ceil{\frac{x_{c}-1}{\tau}} +1 \leq i^* \leq \floor{\frac{y_{c}}{\tau}}, \ceil{\frac{x_{c'}-1}{\tau}} +1 \leq j^* \leq \floor{\frac{y_{c'}}{\tau}}$
such that $p \in I_{i^*}$ and $p' \in I_{j^*}$.
Thus, $\delta_{c,c'} = d(p,p') = d(I_{i^*},I_{j^*})$, and so
\begin{align*}
  \delta_{c,c'} &\leq \min\left( \hat{b}_{i,j},\alpha_c(c'),\alpha_{c'}(c) \right)
  \leq \hat{b}_{i,j}\\
    &= \min_{\ceil{\frac{x_{c}-1}{\tau}} +1 \leq \hat{i} \leq \floor{\frac{y_{c}}{\tau}} , \ceil{\frac{x_{c'}-1}{\tau}} + 1 \leq \hat{j} \leq \floor{\frac{y_{c'}}{\tau}}} \{\hat{b}_{\hat{i},\hat{j}}\}\\
    &\le  \min_{\ceil{\frac{x_{c} -1 }{\tau}} +1 \leq \hat{i} \leq \floor{\frac{y_{c}}{\tau}} , \ceil{\frac{x_{c'} -1}{\tau}} + 1 \leq \hat{j} \leq \floor{\frac{y_{c'}}{\tau}}} \{(1+\varepsilon) d(I_{\hat{i}}, I_{\hat{j}
    })\}\\
  &\leq (1+\varepsilon)d(I_{i^*},I_{j^*})
  = (1+\varepsilon)\delta_{c,c'}.
\end{align*}
To conclude, we have shown that in all cases the algorithm returns an answer that is either exactly $\delta_{c,c'}$ or a $(1+\varepsilon)$ approximation of $\delta_{c,c'}$, which completes the proof of Theorem~\ref{thm: HACDO by ACDO}.

\section{Returning the Points that Define the Distance}\label{app:return points}
 While the statements of Theorems~\ref{thm:ACDO} and~\ref{thm: HACDO by ACDO} state that only the distance is returned, it is rather straightforward to adjust our algorithms to return the actual points that define the (approximate) distance.
 Notice that this feature is needed for the snippets algorithm.
 In this section, we highlight the modifications needed for returning the points themselves.

\paragraph*{Augmenting Theorem~\ref{thm:ACDO}.}
To begin, when answering a query  in the generic algorithm (Section~\ref{sec:overview}), for the case of a query on at least one light set, the algorithm considers all points in one of the light sets, and so assuming that the \NNS{} data structure returns a point, it is straightforward to  return the two points that define the distance.
The more involved case is when both sets are heavy, and in this case we want the matrix $E^*$ to not only store the approximate distance between every pair of heavy colors, but to also store two points, one for each color, that define the returned approximate distance.

Thus, we consider the construction algorithm for $E^*$, described in Section~\ref{sec:construct E*}.
We first augment each $1$ entry in $A$ and each $1$ entry in each $B^{(\ell)}$ with a point of the appropriate color that is in the appropriate block of $S$.
Next, We employ the algorithm of Alon and Naor~\cite{AN96} on the computation of each of the $E^{(\ell)}$ matrices to support storing a witness\footnote{A witness in a Boolean matrix product $C=AB$ for $c_{i,j}=1$ is an index $k$ such that $a_{i,k}=b_{k,j}=1$.} for each $1$ entry.
Thus, whenever the algorithm for constructing $E^*$ (Algorithm 1) reaches Line~\ref{line:assign eij}, the algorithm also retrieves and stores with entry $i,j$ the points defined by the single witness to $e_{i,j}^{(\ell +1)}$.

\paragraph*{Augmenting Theorem~\ref{thm: HACDO by ACDO}.}
The algorithm for \AMCDOCH{} answers queries by either using \RNNS{} queries (which we can assume return the actual point), the \ACDO{} solution of Theorem~\ref{thm:ACDO}, which we discussed above, and a \DRMQ{} query on a matrix $\hat B$ which is an output matrix for $\ACDO{}$ on a specially constructed input (see~\Cref{sec: HACDO by ACDO}).
Adjusting the \DRMQ{} data structure to also store the points is straightforward by treating each entry in $\hat B$ as a key-value pair (the keys used to construct the \DRMQ{} data structure and the value is the two points).

\section{Proof of Theorem~\ref{thm:minplus LB}: CLB for Exact \MCDO{} on an Array}\label{sec:minplus LB}

To prove Theorem~\ref{thm:minplus LB}, we design a reduction from \MPMP{} on values in $[\hat n]$ to \MCDO{}.
To do so, consider an \emph{unbalanced} version of \MPMP{} where $A$ is of size $\hat n \times \hat m$, $B$ is of size $\hat m \times \hat n$, and all of the values are bounded by some positive integer $M$.
It is straightforward to show that for any $x\ge 0$ such that $\hat m = \hat n ^x$, solving unbalanced \MPMP{} with any polynomial time improvement over $(\hat n^2\hat m )^{1-o(1)}$ time is equivalent to solving \MPMP{} on balanced matrices with $\hat m = \hat n$ in $\hat n^{3-\Omega(1)}$ time.
Thus, an algorithm for unbalanced \MPMP{} with values bounded by $M = [\min(\hat n, \hat m)] = \hat{n}$ in time less than $(\hat n^2\hat m )^{1-o(1)}$ would refute the Strong-\APSP{} hypothesis.

\paragraph*{The reduction.}
Our goal is to design a reduction
from the unbalanced \MPMP{} problem with values bounded by $\hat n \ge \min(\hat n, \hat m)$ to \MCDO{}.
In preparation for the proof of Theorem~\ref{thm:CDO LB rand}, we first describe the reduction using positive integer parameter $M$ as the upper bound on the values in the input matrices, and during the analysis we set $M=\hat n$.

For each $a_{i,j}$, the algorithm defines a point $a'_{i,j} = (M - a_{i,j}) + 9Mj$ with color $i$.
For each $b_{i,j}$, the algorithm defines a point $b'_{i,j} = b_{i,j} + 3M +9Mi$ with color $\hat n + j$.
Thus, the largest point defined is at most $O(M \cdot \hat m)$, and there are $2\hat n$ colors.
Notice that there cannot be two entries in the same row (column) of $A$ ($B$) that define the same point.
However, it is possible that $a_{i,j} = a_{i',j}$ for $i\ne i'$, and so $a'_{i,j} = a'_{i',j}$, but the two points have different colors.
A similar phenomenon can happen to points defined by entries of $B$ which share the same row.
Thus, we merge all occurrences of the same point into a single occurrence but colored with all of the colors of the pre-merge occurrences.
The set of points, denoted by $S$, contains at most $2\hat n \hat m$  multi-colored points in a metric defined by an array of size $O(M\hat m)$, and so the algorithm uses the \MCDO{} algorithm to preprocess $S$.
Finally, for each entry $d_{i,j}$ in $D$, the algorithm sets $d_{i,j}\leftarrow \delta_{i,\hat n +j} - 2M$ by performing a query on the \MCDO{} data structure.

\paragraph*{Correctness.}
Suppose $d_{i,j} = a_{i,k^*} + b_{k^*,j}$ for some $k^* \in [ \hat m]$.
For every $k\in [ \hat m]$, by definition we have $b'_{k,j} > a'_{i,k}$ and so
$\abs{a'_{i,k} - b'_{k,j}} = b'_{k,j} - a'_{i,k}
    = b_{k,j} + a_{i,k} + 2M$.
Since $a'_{i,k^*}$ is colored with color $i$ and $b'_{k^*,j}$ is colored with color $\hat n +j$, we have
\begin{equation}\label{eqn:1}
    \delta_{i,\hat n +j}-2M \le \abs{a'_{i,k^*} - b'_{k^*,j}}-2M =  a_{i,k^*} + b_{k^*,j} = d_{i,j}.
\end{equation}

Let $k',k''\in [\hat m]$.
If $k'\ne k''$, then
\begin{align*}
    \abs{a'_{i,k'} - b'_{k'',j}}
    & = \max \big(a'_{i,k'} - b'_{k'',j}, b'_{k'',j} - a'_{i,k'}\big)\\
    &=\max\big( 9M(k'-k'') -2M - a_{i,k'} - b_{k'',j}, b_{k'',j} + a_{i,k'} + 9M(k''-k') + 2M\big) \\
    &\geq 5M.
   \end{align*}
However, if $k'=k''$ then $a_{i,k'} +b_{k',j} +  2M \leq 4M$.
Since $4M < 5M$, we have $$    \delta_{i,\hat{n} +j} = \min_{p \in P(i),p' \in P(\hat{n} +j)}\left\{\abs{p-p'}\right\} = \min_{1\leq \hat{k} \leq \hat{n}}\left\{\abs{a'_{i,\hat{k}} - b'_{\hat{k},j}}\right\}.$$
Thus, there exists some $\hat k\in [\hat m]$ such that $\delta_{i,\hat{n} +j} = \abs{a'_{i,\hat k} - b'_{\hat k,j}} = a_{i,\hat k} + b_{\hat k,j} +2M \ge d_{i,j} +2M$,
and combined with Equation~\ref{eqn:1} we have $d_{i,j} = \delta_{i,\hat n +j}-2M$.

\paragraph*{Lower bound.}
Recall that in our setting $M=\hat n$.
The reduction preprocesses an \MCDO{} data structure on $n\le 2 \hat n\hat m$ points in an array metric of size $O(M\hat m) = O(n)$, and then answers $\hat n^2$ \MCDO{} queries.
If $t_p(n)$ and $t_q(n)$ are  preprocessing and query times, respectively, of the \MCDO{} data structure, then the Strong-\APSP{} conjecture implies that $t_p(\hat n \hat m)+\hat n^2 \cdot t_q(\hat n \hat m) = (\hat n^{2} \hat m)^{1-o(1)}.$

Recall that our goal is to prove that $a+b \ge 2 - o(1)$. Thus, assume towards a contradiction that $a+b = 2 - \Omega(1)$, which, by straightforward algebraic manipulation, implies that $\frac{2a}{a-b} = 1 + \frac{2}{a-b} - \Omega(1).$
Moreover, since the unbalanced \MPMP{} is hard for any choice of $x\ge 0$, we can choose $x = \frac{2}{a-b}-1$ and so $\hat n \hat m = \hat n^{1+x} = \hat n^{\frac{2}{a-b}}$.
Notice that $2+\frac{2b}{a-b} =\frac{2a}{a-b}$
Thus, we have $   t_p(\hat n \hat m)+\hat n^2 \cdot t_q(\hat n \hat m) = O( \hat n^\frac{2a}{a-b} + \hat n ^{2+\frac{2b}{a-b}})
    = O(\hat n^{\frac{2a}{a-b}})= O(\hat n^{1 + \frac{2}{a-b} - \Omega(1)})< (\hat n^2 \hat m)^{1-o(1)}$,
which contradicts the Strong-\APSP{} hypothesis.

\section{Proof of Theorem~\ref{thm:CDO LB rand}:  CLB for Exact \CDO{} on an Array}\label{sec:CDO lb rand}

We first describe a basic reduction algorithm, which is then repeated several times.
The basic algorithm begins by picking uniformly random \emph{offsets} $r_1, ..., r_n$, $s_1, ..., s_n \in [\hat n]$, and then constructs matrices $\hat A$  and $\hat B$ of sizes $\hat n \times \hat m$ and $\hat m \times \hat n$, respectively,  where $\hat a_{i,j} = a_{i,j} + r_i$ and $\hat b_{i,j} = b_{i,j} + s_j$.
Let $\hat D$ be the min-plus product of $\hat A$ and $\hat B$.
It is straightforward from the min-plus definition to verify that $d_{i,j} = \hat d_{i,j} - r_i-s_j$.

The basic algorithm does not compute $\hat D$. Instead,  the algorithm applies on $\hat A$ and $\hat B$ the reduction given in the proof of Theorem~\ref{thm:minplus LB} (see Section~\ref{sec:minplus LB}) to compute a matrix $\tilde D$ with the following two adjustments:
\begin{enumerate}
    \item If the same point appears more than once (which must be with different colors), then remove all copies of the point, regardless of the colors.

\item The computation of each entry $\tilde d_{i,j}$ in $\tilde D$ takes place only if both colors $i$ and $\hat n + j$ have points in $S$. Otherwise we set $\tilde d_{i,j}\leftarrow \infty$.
\end{enumerate}
The first adjustment guarantees that the resulting points are colored points (as opposed to multi-colored points), and so instead of using a \MCDO{}  as in the proof of Theorem~\ref{thm:minplus LB}, the basic algorithm makes use of a \CDO{} for computing each $\tilde d_{i,j}$.

\paragraph*{Repetitions.}
The algorithm repeats the basic algorithm $\alpha$ times, for some parameter $\alpha$ to be set later.
For the $\ell$th iteration, denote the random offsets by
$r^{(\ell)}_1, ..., r^{(\ell)}_n$, $s^{(\ell)}_1, ..., s^{(\ell)}_n$.

Let $\tilde D^{(\ell)}$ be the output of the $\ell$th repetition, and suppose that the number of repetitions is $\alpha$.
Finally, the algorithm computes each $d_{i,j}$ by setting $d_{i,j}\leftarrow \min_{1\le \ell \le \alpha}\{\tilde d^{(\ell)}_{i,j}- r^{(\ell)}_i-s^{(\ell)}_j\} $.

\paragraph*{Correctness.}

First, consider a single execution of the basic algorithm. Notice that due to the adjustments that  the basic algorithm utilizes when applying the reduction of Theorem~\ref{thm:minplus LB}, $\tilde D$ is not necessarily the same as $\hat D$, since some points may be missing from the point set defined by the reduction.
Nevertheless,  from the discussion in the proof of Theorem~\ref{thm:minplus LB} and the property that $d_{i,j} = \hat d_{i,j} - r_i-s_j$, if $ d_{i,j} = a_{i,k^*} +  b_{k^*,j}$ for some $k^*\in [\hat m]$, and the points for both $\hat a_{i,k^*} $ and  $\hat b_{k^*,j}$ were not removed by the adjustment, then $\tilde d_{i,j} - r_i-s_j = \delta_{i, \hat n +j} -2M - r_i-s_j= \hat a_{i,k^*} + \hat b_{k^*,j} - r_i-s_j= \hat d_{i,j}- r_i-s_j = d_{i,j}$, and otherwise, $\tilde d_{i,j}- r_i-s_j \ge d_{i,j}$.

Let $R_{i,j}$ be the event that  the points for both $\hat a_{i,k^*} $ and  $\hat b_{k^*,j}$ were not removed by the adjustment.
In order for the point defined by $\hat a_{i,k^*} $ to be removed by the adjustment, there must be some $i'\neq i$ such that $\hat a_{i,k^*} = \hat a_{i',k^*} $.
In order for such an event to take place, it must be that $r_i-r_{i'} = a_{i',k^*} - a_{i,k^*}$, which for a fixed $k'$ happens with probability $\frac 1 {\hat n}$.
Thus, since all of the $r_1,\ldots, r_{\hat n}$ are independent, the probability that
the point defined by $\hat a_{i,k^*} $ was not removed, which is the same as the probability
that there is no  $i'$ with $\hat a_{i,k^*} = \hat a_{i',k^*}$, is $(1-\frac1 {\hat n})^{\hat n} \approx \frac 1 e$.
Similarly, the probability that the point defined by $\hat b_{k^*,j} $ is not removed is $(1-\frac1 {\hat n})^{\hat n} \approx \frac 1 e$.
Finally, since the events of removing the points defined by
$\hat a_{i,k^*}$ and $\hat b_{k^*,j} $ are independent, we have that $\Pr[R_{i,j}] \approx \frac 1 {e^2}$.

To guarantee that for all $i,j\in [\hat n]$, event $R_{i,j}$  happens in at least one of the repetitions, by applying coupon collector arguments, for high probability guarantees it suffices to set $\alpha = \Theta(\log (\hat n \hat m))$.
It is also possible to continue repeating the basic algorithm until each event $R_{i,j}$ happens at least once, which happens after $\alpha = \Theta(\log (\hat n \hat m))$ repetitions in expectation.

Since for each $i,j \in [\hat n]$, for each $ \ell \in [\alpha]$ we have $\tilde d^{(\ell)}_{i,j} - r^{(\ell)}_i-s^{(\ell)}_j \ge  d_{i,j}$, and for at least one $\ell^* \in [\alpha]$ we have $\tilde d^{(\ell^*)}_{i,j} - r^{(\ell^*)}_i-s^{(\ell^*)}_j =  d_{i,j}$, then  $\min_{1\le \ell \le \alpha}\{\tilde d^{(\ell)}_{i,j} - r^{(\ell)}_i-s^{(\ell)}_j \} = \tilde d^{(\ell^*)}_{i,j} - r^{(\ell^*)}_i-s^{(\ell^*)}_j =  d_{i,j} $.

\paragraph*{Lower bound.}
The process of deriving the lower bound tradeoff is the same as in the proof of Theorem~\ref{thm:minplus LB} (see Section~\ref{sec:minplus LB}) with the following two adjustments: (1) in our case $M = 2\hat n$, and (2) we perform $\Theta(\log n)$ repetitions.
Both of these adjustments do not change the exponents in the polynomials of the running times, and so they do not affect the lower bound tradeoff curve.

\section{Conclusions and Open Problems}
We have shown the existence of  FMM based algorithms for both \ACDO{} and the $(1+\varepsilon)$-approximate snippets problem which, assuming $\omega =2$, are essentially optimal.
Moreover, we proved CLBs for exact version of \CDO{}, implying that the exact versions of \CDO{} and the snippets problem are strictly harder than their approximate versions.

We remark that one immediate and straightforward way to improve our algorithms if $\omega >2$ is to apply fast rectangular matrix multiplication (\cite{ADWVXZ25}).
However, we chose not to describe such improvements since they are mostly technical and do not add any additional insight to the problems that we address.
Moreover, we remark that it is straightforward to  adapt our \CDO{} algorithms to  return the two points (one of each color) that define the distance in addition to the actual distance.
Our work leaves open the task of designing improved algorithms for more general metrics such as higher dimensional Euclidean space.

\paragraph*{Acknowledgment}
The authors thank an anonymous reviewer who suggested some ideas for tightening the lower bounds proved in Theorems~\ref{thm:minplus LB} and~\ref{thm:CDO LB rand}.
\bibliographystyle{plainurl}
\bibliography{bib}

\end{document}